\documentclass[12pt]{article}

\usepackage{fullpage}

\usepackage{enumerate}

\usepackage{epsfig}

\usepackage{amsthm}
\usepackage{amssymb, amsmath, latexsym}

\theoremstyle{plain}
\newtheorem{theorem}{Theorem}
\newtheorem{corollary}{Corollary}
\newtheorem{lemma}{Lemma}

\theoremstyle{definition}

\newcommand{\cancel}[1]{}

\newcommand{\eps}{\varepsilon}

\newcommand{\mE}{\mathcal E}
\newcommand{\mX}{\mathcal X}

\newcommand{\mS}{\mathcal S}

\newcommand{\mQ}{\mathcal Q}

\DeclareMathOperator{\tr}{tr}

\DeclareMathOperator{\ext}{ext}
\DeclareMathOperator{\guess}{guess}
\DeclareMathOperator{\ReSamp}{ReSamp}

\newcommand{\bra}[1]{{\langle {#1}|}}
\newcommand{\ket}[1]{{|{#1}\rangle}}

\newcommand{\ketbra}[2]{{\ket{#1}{\bra{#2}}}}

\title{Bitwise Quantum Min-Entropy Sampling and \\ New Lower Bounds for Random Access Codes}

\begin{document}

\author{J{\"u}rg Wullschleger\\
\textit{DIRO, Universit\'e de Montr\'eal, Quebec, Canada}\\
\textit{McGill University, Quebec, Canada}
}

\maketitle

\begin{abstract}
 \emph{Min-entropy sampling} gives a bound on the min-entropy of a randomly chosen subset of a string, given a bound on the min-entropy of the whole string. K\"onig and Renner showed a min-entropy sampling theorem that holds relative to quantum knowledge. Their result achieves the optimal rate, but it can only be applied if the bits are sampled in block, and only gives weak bounds for non-smooth min-entropy.

We give two new quantum min-entropy sampling theorems that do not have the above weaknesses. The first theorem shows that the result by K\"onig and Renner also applies to bitwise sampling, and the second theorem gives a strong bound for the non-smooth min-entropy.

Our results imply strong lower bounds for $k$-out-of-$n$ random access codes. While previous results by Ben-Aroya, Regev, and de~Wolf showed that the decoding probability is exponentially small in $k$ if the storage rate is smaller than $0.7$, our results imply that this holds for any storage rate strictly smaller than $1$, which is optimal.
\end{abstract}

\section{Introduction}

Let us assume that two players share a long string $x \in \{0,1\}^n$, over which an adversary has only partial knowledge. They would like to get a key, over which the adversary has almost no knowledge. Since the string is long, using a 2-universal hash function or, more generally, a normal strong extractor would be inefficient and hence impractical. Vadhan showed in \cite{Vadhan04} that the two players can instead first randomly sample a relatively small substring $x' \in \{0,1\}^{k}$ of $x$, and then apply an extractor to $x'$. This works because with high probability, the string $x'$ will have almost $\frac {k} n \cdot t$ bits of min-entropy, if the min-entropy of $x$ is at least $t$.
K\"onig and Renner showed in \cite{KoeRen07} that this holds works in the more general case where the adversary has quantum information about $x$. Again, with high probability the string $x'$ will have almost $\frac {k} n \cdot t$ bits of quantum min-entropy. 

Related to these results are lower bounds for \emph{random access codes}. 
This is an encoding of $n$ classical bits into $m<n$ qubits, such that from the encoding, a randomly chosen subset of size $k$ can be guessed with probability at least $p$. The first lower bound was given for the case where $k = 1$ by Ambainis, Nayak, Ta-Shma and Vazirani in \cite{ANTV98}. It was later improved by Nayak in  \cite{Nayak99} to $m \geq (1 - H(p))n$, where $H( \cdot )$ is the binary entropy function. For the general case where $k \geq 1$, a lower bound was presented by Ben-Aroya, Regev, and de~Wolf in \cite{BeReWo08}. They showed that for any $\eta > 2 \ln 2$ there exists a constant $C_{\eta}$ such that
\[ p \leq C_{\eta} \left ( \frac 1 2 + \frac 1 2 \sqrt{ \frac{\eta m} n } \right ) ^k\;.\]
It implies that if $m < n / (2 \ln 2) \approx 0.7 n$, then $p \leq 2 ^{-\Omega(k)}$.
In the same work they also showed lower bounds for a variant of random access codes called \emph{XOR-random access codes}, where the player is asked to guess the XOR of a random subset of size $k$. De and Vidick presented in \cite{DeVid10} lower bounds for \emph{functional access codes}, which is a generalization of XOR-random access codes where the player is asked to guess the output of a function with binary output chosen from a bigger set.

The result in \cite{Vadhan04} implies a classical lower bound for $k$-out-of-$n$ random access codes. In principle, this would also be possible in the quantum setting, as the min-entropy is defined as minus the logarithm of the guess probability. Unfortunately, the results by K\"onig and Renner are not general enough to do that, because they require the sampling to be done in blocks.

\cite{BeReWo08} showed that lower bounds for $k$-out-of-$n$ random access codes imply lower bounds for the one-way communication complexity of $k$ instances of the disjointness problem.

\subsection{Contribution}

In this work we give two new results for quantum min-entropy sampling.

First, we show in Theorem \ref{thm:bitwise} in Section~\ref{sec:bitwise} that the bounds given in Corollary 6.19 and Lemma 7.2 in \cite{KoeRen07} also apply to the case where the sample is chosen bitwise, instead of (recursively) in blocks. This result simplifies some protocols\footnote{For example, it allows that the simpler and more intuitive Protocol~2' in \cite{KoWeWu09} can be proved secure, instead of the more complicated Protocol~2.} as it eliminates an artificial extra step where the bits have to be grouped in blocks.

Second, building on previous results given in \cite{BeReWo08} and \cite{DeVid10}, in Section~\ref{sec:mainthm} we will give a new quantum sampling theorem (Theorem \ref{thm:main}). The proof of Theorem \ref{thm:main} is much simpler than the min-entropy sampling results in \cite{KoeRen07}, and give stronger bounds for non-smooth min-entropy. It implies the following corollary.

\begin{corollary} \label{cor:main}
Let a cq-state $\rho_{XQ}$ be given, where $X \in \{0,1\}^n$. Let $T$ be a random subset of $[n]$ of size $k$. If for a constant $c \in [0,1]$ we have $H_{\min}(X \mid Q)_\rho \geq cn$,
then
\[ H_{\min}(X_T \mid T Q)_\rho \geq \frac{H^{-1}(c/2)}6  k  - 5\;.\]
\end{corollary}

Corollary~\ref{cor:main} immediately implies the following bound for random access codes.

\begin{corollary} \label{cor:RAC2}
Let $\eps>0$ be a constant. For any $k$-out-of-$n$ random access code where the storage is bounded by $m \leq (1-\eps)n$, the success probability is at most $2^{-\Omega(k)}$.
\end{corollary}

As the results in \cite{BeReWo08}, Corollary~\ref{cor:RAC2} generalizes the bound given by Nayak to the case where $k \geq 1$. But while the results in \cite{BeReWo08} require that $m < 0.7 n$, our results imply that the  success probability decreases exponentially in $k$ even if $m$ is close to $n$.




Together with Lemma 8 in \cite{BeReWo08}, Corollary \ref{cor:RAC2} implies a strong lower bound for the one-way communication complexity of $k$ independent instances of the disjointness problem.

\section{Preliminaries}

The \emph{binary entropy function} is defined as
$H(x) := - x \log x - (1-x) \log (1-x)$
for $x \in [0,1]$, where we use the convention $0 \log 0 = 0$. For $y \in [0,1]$, let $H^{-1}(y)$ be the value $x \in [0,\frac12]$ such that $H(x) = y$. The \emph{Hamming distance} $d_H$ between two strings is defined as the number of bits where the two strings disagree.
We use the notion $[n] := \{1, \dots, n\}$. The substring of $x \in \{0,1\}^n$ defined by the set $s \subset [n]$ is denoted by $x_s$. 

Let $\rho_{XQ}$ be a cq-state of the form
$\rho_{XQ} = \sum_{x} p_x \ketbra{x}{x} \otimes \rho^x_Q$.
The \emph{conditional min-entropy} is defined as
\[ H_{\min}(X \mid Q)_\rho := - \log P_{\guess} (X \mid Q)_\rho\;,\]
where 
\[P_{\guess} (X \mid Q)_\rho := \max_{\mE} \sum_{x \in \mX} P_{X}(x) \tr(E_x\rho_x)\;.\]
The maximum is taken over all POVMs $\mE = \{ E_x\}_{x \in \mX}$ on $\mQ$. $P_{\guess} (X \mid Q)_\rho$ is therefore the probability to correctly guess $X$ by measuring system $Q$. The equivalence of this definition of $H_{\min}$ with the definition used in \cite{KoeRen07} has been shown in \cite{KoReSc08} in Theorem 1.
The \emph{statistical distance} $D(\rho,\phi)$ between two states $\rho$ and $\phi$ is defined as
\[ D(\rho,\phi) = \max_{\mE} \ | \tr(E_1 \rho) - \tr(E_1 \phi) | \;, \]
where we maximize over all POVMs $\mE = \{ E_x\}_{x \in \{0,1\}}$. $D(\rho,\phi)$ is therefore the maximal probability to distinguish $\rho$ and $\phi$ by a measurement. It can be shown that $D(\rho,\phi) = \frac 1 2 \| \rho - \phi\|_1$.

\begin{lemma}
Let $\rho_{X Q}$ be a cq-state where $X$ is binary and let $\tau_X$ be the fully mixed state. Then $D(\rho_{XQ},\tau_X \otimes \rho_Q) \leq \eps$ implies that $P_{\guess} (X \mid Q)_\rho \leq \frac12 + \eps$.
\end{lemma}

\begin{proof}
Let us assume that there exists a POVM $\mE$ on $\mQ$ that can guess $X$ with a probability bigger than $\frac12 + \eps$. We define a POVM $\mE'$ on $\mX \otimes \mQ$ in the following way: We measure $Q$ using $\mE$ and XOR the output with $X$.
We have $\tr(E'_1 \rho_{X E} ) < \frac12 - \eps$ and $\tr(E'_1 (\tau_X \otimes \rho_Q) ) = \frac 1 2$. Hence  $D(\rho_{XQ},\tau_X \otimes \rho_Q) > \eps$, which contradicts the assumption.
\end{proof}

\begin{lemma}[Chernoff/Hoeffding] \label{lem:chernoff1}
Let $P_{X_0\dots X_n} = P_{X}^n$ be a product distribution
with $X_i \in [0,1]$.  Let
$X := \frac 1 n \sum_{i=0}^{n-1} X_i$, and $\mu = E[X]$. Then, for any $\eps > 0$, 
$\Pr\left[ X \leq \mu - \eps \right ] \leq e^{-2n\eps^2}$.
\end{lemma}

\section{Bitwise Sampling from Blockwise Sampling} \label{sec:bitwise}

In this section we show that the min-entropy sampling results from \cite{KoeRen07}, which require blockwise sampling, also imply the same bounds for uniform bitwise sampling.

The following theorem is the statement of Corollary 6.19 in \cite{KoeRen07} for uniform sampling. Here $H^\eps_{\min}$ is the \emph{smooth min-entropy}, and $H_0$ the \emph{R\'enyi $0$-entropy}. The definitions of these entropies and their properties can be found in Section 5 in \cite{KoeRen07} or Chapter 3 in \cite{Renner05}.

\begin{theorem}[\cite{KoeRen07}] \label{thm:KR1}
Let $\rho_{XQ}$ be a cq-state where $X = (X_1, \dots, X_n) \in \mX^n$. Let $S \subset [n]$ be chosen uniformly at random among all subsets of size $r$. Assume that $\kappa = \frac{n}{r \log |\mX|} \leq 0.15$. Then
\[ \frac{H^{\eps}_{\min}(X_{S} \mid S, Q )}{H_0(X_{\mS})} \geq \frac{H_{\min}(X \mid Q )}{H_0(X)} - 3 \xi - 2 \kappa \log 1/\kappa \;, \]
where $\eps = 2 \cdot 2^{- \xi n \log|\mX|} + 3 e^{-r \xi^2/8}$
\end{theorem}

The statement says that with high probability, the min-entropy rate of a random subset is almost as big as the min-entropy rate of the whole string. To achieve the required condition $n \leq 0.15 \cdot r \log |\mX|$ (for example if $X$ is a bit string), $X$ might have to be grouped into blocks first.
But as pointed out in \cite{BeReWo08}, even then the statement cannot be applied if we want to sample a subset that is smaller than the square-root of the total length of the bit string. 

To overcome this problem, \cite{KoeRen07} proposed a \emph{recursive} application of Theorem~\ref{thm:KR1}. The following theorem is Lemma 7.2 in \cite{KoeRen07}. See Section 7 in \cite{KoeRen07} for the exact definition of the sampling algorithm $\ReSamp(X,f,r,S)$. 

\begin{theorem}[\cite{KoeRen07}] \label{thm:KR2}
Let $\rho_{XQ}$ be a cq-state where $X$ is a $n$-bit string. Let $n$, $f$ and $r$ be such that $n^{(3/4)^f} \geq r^4$. Let $S$ be a string of uniform random bits, and let $Z = \ReSamp(X,f,r,S)$. Then $Z$ is a $n^{(3/4)^f}$-bit substring of $X$, with
\[ \frac{H^{\eps}_{\min}(Z \mid S,Q)}{H_0(Z)} \geq \frac{H_{\min}(X \mid Q )}{H_0(X)} - 
5 f \frac{\log r}{r^{1/4}} \;, \]
where $\eps = 5 f \cdot 2^{-\sqrt{r}/8}$.
\end{theorem}

Since bitwise sampling is generally better than blockwise sampling, it seems that the results of both Theorem~\ref{thm:KR1} and \ref{thm:KR2} should also hold if the subset is sampled bitwise uniformly. The following theorem shows that this is indeed the case.

\begin{theorem}\label{thm:bitwise}
The bound of Theorem~\ref{thm:KR1} and \ref{thm:KR2} also apply if the sample is chosen bitwise uniformly.
\end{theorem}

\begin{proof} 
Let $k, n \in \mathbb{N}$, were $k < n$. Let $\rho_{X Q}$ be a cq-state where $X \in \{0,1\}^{n}$. Let $S \subset [n]$ be chosen uniformly at random from all subset of size $k$ and let $T \subset [n]$ be a random subset of size $k$ chosen according to a given distribution $P_T$. Let $\Pi$ a permutation chosen uniformly at random, but such that it maps all elements in $S$ into $T$.
Strong subadditivity (Theorem~3.2.12 in \cite{Renner05}) implies
\begin{align*}
 H^\eps_{\min}(X_{S} \mid S, Q )
 &\geq H^\eps_{\min}(X_{S} \mid  S, \Pi, Q ) \\
 &= H^\eps_{\min}(\Pi(X)_{T} \mid  T, \Pi, Q ) \;.
\end{align*}
Note that from $(S,\Pi)$ it is possible to calculate $(T,\Pi)$, and vice-versa.
Furthermore, since $\Pi$ is chosen independent of $\rho_{XQ}$, we have
 \[H^\eps_{\min}(\Pi(X) \mid \Pi, Q) = H^\eps_{\min}(X \mid \Pi, Q) = H^\eps_{\min}(X \mid Q)\;.\]
Since $S$ was chosen uniformly and independent of $T$ and $\rho_{XQ}$, $\Pi$
is independent of $T$ and $\rho_{XQ}$.
Setting $Q' := (Q,\Pi)$, we can apply Theorem~\ref{thm:KR1} or \ref{thm:KR2} to the state $\rho_{\Pi(X) Q'}$, and get a bound on $H^\eps_{\min}(\Pi(X)_{T} \mid  T, \Pi, Q )$, which then directly implies the same bound for $H^\eps_{\min}(X_{S} \mid S, Q )$.
\end{proof}

\section{A Sampling Theorem from Quantum Bit Extractors} \label{sec:mainthm}

In this section we give a new min-entropy sampling theorem using a completely different approach than \cite{KoeRen07}. Our proof has two steps. First, we show a bound on the guessing probability of
the XOR of a randomly chosen substring of $X$ using results from \cite{DeVid10}, which are based on \emph{strong quantum extractors}. Second, we will show that this implies a bound on the guessing probability of a randomly chosen substring of $X$. To show this we use a similar approach as the proof of Theorem 2 in \cite{BeReWo08}.


 A function $\ext: \{0,1\}^n \times \{0,1\}^d \rightarrow \{0,1\}^m$ is a \emph{$(\ell,\eps)$-strong extractor against quantum adversaries}, if for all states $\rho_{XQ}$ that are classical on $X$ with $H_{\min}(X \mid Q)_\rho \geq \ell$ and for a uniform seed $R$, we have
$D(\rho_{\ext(X,R)RQ},\tau_U \otimes \rho_R \otimes \rho_Q) \leq \eps$,
where $\tau_U$ is the fully mixed state.
A strong \emph{classical} extractor is the same, but with a trivial system $Q$. If $m=1$, we call it a bit-extractor. K\"onig and Terhal showed in \cite{KoeTer08} that any classical bit-extractor is also a quantum bit-extractor.

\begin{theorem}[Theorem III.1 in \cite{KoeTer08}] \label{thm:KT}
 Any $(\ell, \eps)$-strong bit-extractor is a $(\ell + \log 1/\eps, 3 \sqrt{\eps})$-strong bit-extractor against quantum adversaries.
\end{theorem}

One way to construct a strong bit-extractor is to use a \emph{$(\eps, \delta, L)$-approximately list-decodable code}, which is a code $C: \{0,1\}^n \rightarrow \{0,1\}^{m}$ where for every $c' \in \{0,1\}^{m}$ there exist $L$ strings $c_1, \dots, c_L \in \{0,1\}^n$, such that for any string $x \in \{0,1\}^n$ satisfying $d_H(c', C(x)) < (\frac12 - \eps) m$, there exists an $i \in \{1,\dots,L\}$ such that $d_H(c', c_i) \leq \delta m$.
From a code $C: \{0,1\}^n \rightarrow \{0,1\}^{2^t}$, we can build a bit-extractor $\ext: \{0,1\}^n \times \{0,1\}^t \rightarrow \{0,1\}$ as
$\ext(x,y) := C(x)_y$,
where $C(x)_y$ is the $y$th position of the codeword $C(x)$.

\begin{lemma}[Claim 3.7 in \cite{DeVid10}] \label{lem:37}
Let $\delta \in [0,\frac12]$. An extractor build from a $(\eps,\delta,L)$-approximately list-decodable code $C: \{0,1\}^n \rightarrow \{0,1\}^{2^t}$  is a $(\ell,\eps)$-strong classical bit-extractor for 
$\ell > H(\delta)n + \log L + \log 2/\eps$.
\end{lemma}

The $(n,k)$-XOR-code over strings of length $n$ is the code where the string $x$ gets encoded into a string of size $\binom n k$ where each bit is the XOR of a subset of $x$ of size $k$.

\begin{lemma}[Lemma 42 in \cite{ImJaKa06}, adapted in \cite{DeVid10}, Lemma 3.11] \label{lem:42}
 For $\eps > 2 k^2 / 2^n$, the $(n,k)$-XOR-code is a $(\eps, \frac 1 k \ln \frac 2 \eps, 4/\eps^2)$-approximately list-decodable code.
\end{lemma}


Combining Lemmas~\ref{lem:37} and \ref{lem:42} with Theorem~\ref{thm:KT}, we get the following lemma.

\begin{lemma} \label{lem:3}
 Let $\eps > 2 k^2 / 2^n$ and $k \geq 2 \ln \frac 2 \eps$. The extractor build from the $(n,k)$-XOR-code implies a $(\ell, 3 \sqrt{\eps})$-strong bit-extractor against quantum adversaries for
\[ \ell > H \Big (\frac 1 k \ln \frac 2 \eps \Big ) n + 4 \log \frac 1 \eps + 3\;.\]
\end{lemma}

\begin{proof}
Using Lemma~\ref{lem:37} and \ref{lem:42}, the $(n,k)$-XOR-code implies a $(\ell,\eps)$-strong classical bit-extractor for
\begin{align*}
 \ell > H \Big (\frac 1 k \ln \frac 2 \eps \Big ) n + \log \frac 4 {\eps^2} + \log \frac 2 \eps
   = H \Big (\frac 1 k \ln \frac 2 \eps \Big ) n + 3 \log \frac 1 \eps + 3\;.
\end{align*}
The statement follows from Theorem~\ref{thm:KT}.
\end{proof}

From Lemma \ref{lem:3} follows that that if a string $X$ can only be guessed from $Q$ with probability at most $2^{-\ell}$, i.e., $H_{\min}(X \mid Q) \geq \ell$, then the XOR of a random subset of size $k$ can be guessed with probability at most $1/2 + 3 \sqrt{\eps}$. 

The following lemma gives a bound on the probability to guess a whole substring, given bounds on the probability to guess the XOR of substrings. It has been proven as a part of Theorem 2 in \cite{BeReWo08}. For clarity, we include the proof here.

\begin{lemma}[part of Theorem 2 in \cite{BeReWo08}] \label{lem:BRW}
Let $\rho_{X Q}$ be a cq-state where $X \in \{0,1\}^n$ and let $p_0, \dots, p_k > 0$ be upper bounds on the probability to guess the XOR of a random subset of $X$ of size $j$ given $Q$ and the subset. 
Then the probability to guess a random subset of $X$ of size $k$ from $Q$ and the subset is at most
\[ \frac 1 {2^k} \sum_{j=0}^k \binom k j (2 p_j - 1)\;.\]
\end{lemma}

\begin{proof}
Let $P_T$ be the uniform distribution among all subsets of $[n]$ of size $k$, and let $t$ be distributed according to $P_T$. Let $P_{S \mid T=t}$ be the distribution that chooses a random subset of $t$. This defines the joint distribution $P_{ST}$, as well as the distributions $P_S$ and $P_{T \mid S=s}$. Let 
\[ P_J(j) :=  \frac {1} {2^k} \binom k j\;,\]
for $j \in \{0, \dots, k\}$. $P_J(j)$ is the probability that the subset $s$ has size $j$.
We have $P_S(s) = \sum_j P_J(j) P_{S \mid J=j}(s)$, where $P_{S \mid J=j}(s)$ is the uniform distribution over the subsets of $[n]$ of size $j$.

For any $t$, let $\mE_t$ be a POVM on $\mQ$ that guesses $X$ for the subset $t$, for $t$ chosen according to $P_T$. For any $t$, this defines a distribution $P_{W \mid T=t}$ over error-strings $w \in \{0,1\}^k$, where $w=0^k$ means that the guess was correct.
For $s \subset [k]$, let
\begin{align*}
 Q_{S \mid T=t}(s) := \frac{1}{2^k} \sum_{w \in \{0,1\}^k} P_{W \mid T=t}(w) \chi_s(w)\;,
\end{align*}
where $\chi_s(x) := (-1)^{x \cdot s}$, i.e., it is the parity of the bits of $w$ indexed by $s$.
$Q_{S \mid T=t}$ is the Fourier-transform of $P_{W \mid T=t}$, so we also have
\begin{align*}
 P_{W \mid T=t}(w) = \sum_{s \subseteq [k]} Q_{S \mid T=t}(s) \chi_s(w)\;.
\end{align*}
Using all subsets of $t$ as the domain of $s$, and since $\chi_s(0) = 1$ for all $s$ we can write
\begin{align*}
P_{W \mid T=t}(0)
 = \sum_{s \subseteq t} Q_{S \mid T=t}(s) \chi_s(0)
 = \sum_{s \subseteq t} \frac{1}{2^k} \sum_{w \in \{0,1\}^k} P_{W \mid T=t}(w) \chi_s(w)\;.
\end{align*}

Let $p$ be the maximal probability to guess a subset $t$ distributed according to $P_T$.
Note that for $s \subseteq t$, we have $P_{S \mid T=t}(s) = 2^{-k}$. We get
\begin{align*}
p
& = \sum_t P_T(t) P_{W \mid T=t}(0) \\
& =  \sum_t P_T(t) \sum_{s \subseteq t} \frac{1}{2^k} \sum_w P_{W \mid T=t}(w) \chi_s(w)  \\
& = \sum_t P_T(t) \sum_{s \subseteq t} P_{S \mid T=t}(s) \sum_w P_{W \mid T=t}(w) \chi_s(w)  \\
& =  \sum_{t,s} P_{TS}(t,s) \sum_w P_{W \mid T=t}(w) \chi_s(w)  \\
& =  \sum_{j} P_J(j) \sum_{s} P_{S \mid J=j}(s) \sum_{t} P_{T \mid S=s}(t) \sum_w P_{W \mid T=t}(w) \chi_s(w)\;.
\end{align*}
Given a set $s$ of size $j$, let us apply the following algorithm to guess the XOR of the subset $s$ of $X$. We first sample $t$ according to $P_{T \mid S=s}(t)$, then apply the POVM $\mE_t$ to $Q$, and then output the XOR of the bits in $s$ of the outcome. 
The probability that this algorithm guesses correctly the XOR of a randomly chosen subset of size $j$ is
\[ \sum_s P_{S \mid J=j} \sum_{t} P_{T \mid S=s}(t) \sum_w P_{W \mid T=t}(w) \frac{1+\chi_s(w)} 2\;.\]
Since this must be upper bounded by $p_j$, it follows that
\[\sum_s P_{S \mid J=j}  \sum_{t} P_{T \mid S=s}(t) \sum_w P_{W \mid T=t}(w) \chi_s(w) \leq 2 p_j - 1\;.\]
So we get
\begin{align*}
 p  \leq  \sum_{j} P_{J}(j) (2 p_j - 1) =  \sum^k_{j=0} \frac {1} {2^k} \binom k j (2 p_j - 1)\;.
\end{align*}

\end{proof}

We can now use Lemmas \ref{lem:3} and \ref{lem:BRW} to proof our main result.

\begin{theorem} \label{thm:main}
Let a cq-state $\rho_{XQ}$ be given, where $X \in \{0,1\}^n$. Let $T$ be a random subset of $[n]$ of size $k$. If $\log \frac{1}{p} \leq k/12 - 5$ and
\[
 H_{\min}(X \mid Q)_\rho \geq 
H \left (\frac 6 k \log \frac {17} {p} \right ) n + 8 \log \frac {12} {p} + 3\;,
\]
then $H_{\min}(X_T \mid T Q)_\rho \geq \log \frac 1 p$.
\end{theorem}

\begin{proof}
From $\log \frac{1}{p} \leq k/12 - 5$ follows that $12 \log(17/p) \leq k$ and hence also $17 \ln (17/p) \leq k$.
Since $k \leq n$ and $5k/12 \geq \log(17k)-5$, it follows also that 
\[\log \frac 1 p \leq \frac k {12} - 5 \leq \frac k 2 - \log(17 k) \leq \frac n 2 - \log(17 k)\]
and hence $p^2 \geq 288 \cdot k^2/2^n$.
For $j \in \{0, \dots, n\}$, let $p_j$ be the guess probability of the XOR for random subsets of size $j$. From Lemma~\ref{lem:BRW} follows that
\begin{align*}
 P_{\guess}(X_T \mid T Q)_\rho
 &\leq \frac 1 {2^k} \sum_{j=0}^k \binom k j (2 p_j - 1) \\
 &\leq  \frac 1 {2^k} \sum_{j=0}^{k/4} \binom k j + \max_{j' \in [k/4+1,k]} (2 p_{j'} - 1) \cdot \frac 1 {2^k} \sum_{j=k/4+1}^k \binom k j  \\
 &\leq  \frac 1 {2^k} \sum_{j=0}^{k/4} \binom k j + \max_{j' \in [k/4+1,k]} (2 p_{j'} - 1)\;.  
\end{align*}
We have
\[  \sum_{j=0}^{k/4} \frac 1 {2^k} \binom k j = \Pr \Big [ J \leq k/4 \Big ]\;, \]
where $J = \sum_{i \in [k]} J_i$ and $J_i$ are independent and uniform on $\{0,1\}$.
From Lemma~\ref{lem:chernoff1} follows 
that 
\[ \Pr [ J \leq k/4 ] \leq \exp(- k / 8) \leq p/2\;,\]
since $k \geq 17 \ln \frac {17} p > 8 \ln \frac {2} p$.
Let $\eps := p^2/144$. Since $k \geq 17 \ln \frac {17} p$, we have
\begin{align*}
\frac 1 2 > \frac 8 k \ln \frac {17} {p}
> \frac 4 k \ln \frac {288} {p^2}
= \frac 4 k \ln \frac {2} {\eps}\;,
\end{align*}
and hence
\begin{align*}
 H_{\min}(X \mid Q)_\rho 
&\geq H \Big (\frac 4 k \ln \frac 2 \eps \Big ) n + 4 \log \frac 1 \eps + 3\;.
\end{align*}
From $p^2 \geq 288 \cdot k^2/2^n$ follows that $\eps \geq 2k^2/2^n \geq 2 (k/4)^2/2^n$.
Lemma~\ref{lem:3} implies that
\[\max_{j' \in [k/4+1,k]} (2 p_{j'} - 1) \leq 6 \sqrt{\eps}= p/2\;.\]
The statement follows from the definition of $H_{\min}$.
\end{proof}

\begin{proof}[Proof of Corollary~\ref{cor:main}]
Let $p := 2^{-H^{-1}(c/2)k/6 - 5}$, which implies
\[ H \left (\frac 6 k \log \frac {17} {p} \right ) \leq \frac{c} 2\;.\]
From $H^{-1}(c/2) \leq \frac12$ follows that
\[\log \frac 1 p = \frac{H^{-1}(c/2)} 6 k - 5 \leq \frac{k} {12} - 4\;.  \]
Since $n \geq k$ and $\frac 1 2 \geq x \geq H^{-1}(x)$, we have
\[\log \frac 1 p = \frac{H^{-1}(c/2)} 6 k - 5 \leq H^{-1}(c/2) \cdot \frac n 6 - 5 \leq \frac{c/2} {2} \cdot \frac n 6 - 5 \leq \frac{c n} {24} - 5\;,\]
which implies
\[ 8 \log \frac {12} p + 3 = 8 \log \frac {1} p + 8 \log(12) + 3 \leq \frac {cn}{3} - 40 + 32 + 3 \leq \frac{cn}{2}\]
Hence,
\[ cn \geq H \left (\frac 6 k \log \frac {17} {p} \right )n + 8 \log \frac {12} p + 3\;.\]
The statement follows from Theorem~\ref{thm:main}.
\end{proof}

\section{Lower Bounds for Random Access Codes}

Corollary~\ref{cor:main} directly implies a lower bound for random access codes:
if we choose the string $X \in \{0,1\}^n$ uniformly and the quantum system $Q$ has at most $m \leq (1-\eps)n$ qubits, then by Proposition~2' in \cite{KoeTer08}, we have $H_{\min}(X \mid Q) \geq \eps n$. Corollary~\ref{cor:RAC2} follows.

Theorem \ref{thm:KR1} or \ref{thm:KR2} in combination with Theorem \ref{thm:bitwise} can be used to give a bound for random acces codes, since $H^\eps_{\min}(X \mid Q) \geq \ell$ implies $P_{\guess}(X \mid Q) \geq 2^{-\ell} + \eps$. But even though the bounds given in Theorem~\ref{thm:KR1} and \ref{thm:KR2} are almost tight for the \emph{smooth} min-entropy, they only give weak bounds on the min-entropy, since the smoothness error $\eps$ is relatively big:
in Theorem~\ref{thm:KR1} (using $\mX = \{0,1\}^b$), we sample from $n' := n/b$ block a subset of $r := k/b$ blocks. Since $\xi \leq 1$ and $\kappa = \frac{n'}{r b} = \frac{n}{k b} \leq 0.15$, we have $b \geq \frac{n}{0.15 k}$. Therefore
\[ \eps \geq 3 e^{-r \xi^2/8} = 3 e^{-k \xi^2/(8 b) } \geq 3 e^{-0.15 k^2 \xi^2/(8 n) } \geq 2^{-O(k^2/n)}\;.\]
In Theorems \ref{thm:KR2}, it is required that
$k \geq r^4$, which implies that
\[\eps = 5 f \cdot 2^{-\sqrt{r} / 8} \geq 5 f \cdot 2^{-\sqrt[8]{k} / 8} = 2^{-O(\sqrt[8]{k})}\;.\]
Therefore, Theorem \ref{thm:KR1} and \ref{thm:KR2} in combination with Theorem~\ref{thm:bitwise} can only provide us with weak bounds for random access codes.

\section{Open Problems}


Both our sampling results only apply to the case where the sample is chosen uniformly. It would be interesting to know if they can be generalized to other sampling strategies.

\subparagraph*{Acknowledgements:} I thank Robert K\"onig, Thomas Vidick and Stephanie Wehner for helpful discussions. This work was funded by the U.K. EPSRC grant EP/E04297X/1 and the Canada-France NSERC-ANR project FREQUENCY. Most of this work was done while I was at the University of Bristol.


\end{document}